\documentclass[letterpaper]{article} % DO NOT CHANGE THIS
\usepackage[submission]{aaai24}  % DO NOT CHANGE THIS
\usepackage{times}  % DO NOT CHANGE THIS
\usepackage{helvet}  % DO NOT CHANGE THIS
\usepackage{courier}  % DO NOT CHANGE THIS
\usepackage[hyphens]{url}  % DO NOT CHANGE THIS
\usepackage{graphicx} % DO NOT CHANGE THIS
\usepackage{natbib}  % DO NOT CHANGE THIS AND DO NOT ADD ANY OPTIONS TO IT
\usepackage{caption} % DO NOT CHANGE THIS AND DO NOT ADD ANY OPTIONS TO IT
\usepackage{amsmath,amssymb,amsthm}
\usepackage{amsfonts}
\usepackage{soul}
\usepackage[ruled,linesnumbered,vlined]{algorithm2e}
\usepackage{tikz}
\usepackage{pgfplots}
\pgfplotsset{compat=1.17}
\usetikzlibrary{patterns}
\usepackage{pgfplotstable}
\usepackage{thm-restate}
\pgfplotsset{
	name nodes near coords/.style={
		every node near coord/.append style={
			name=#1-\coordindex,
			alias=#1-last,
		},
	},
	name nodes near coords/.default=coordnode
}

\usepackage{newfloat}
\usepackage{listings}
\DeclareCaptionStyle{ruled}{labelfont=normalfont,labelsep=colon,strut=off} % DO NOT CHANGE THIS
\newcommand{\bR}{{\mathbb{R}}}
\newcommand{\bX}{\mathbf{X}}
\newcommand{\bc}{\mathbf{c}}
\newcommand{\bd}{\mathbf{d}}

\renewcommand{\sc}{{\rm sc}}

\newtheorem{claim}{Claim}[section]
\newtheorem{theorem}{Theorem}[section]
\newtheorem{corollary}[theorem]{Corollary}
\newtheorem{lemma}[theorem]{Lemma}

\newtheorem{example}[theorem]{Example}
\newtheorem{proposition}[theorem]{Proposition}

\newtheorem{definition}[theorem]{Definition}
\newtheorem{remark}[theorem]{Remark}

\newenvironment{proofof}[1]{{\vspace*{5pt} \noindent\bf Proof of #1:  }}{\hfill\rule{2mm}{2mm}\vspace*{5pt}}

\DeclareMathOperator*{\argmin}{argmin}

\makeatletter
\newcommand{\removelatexerror}{\let\@latex@error\@gobble}
\makeatother

\title{On the Existence of EFX (and Pareto-Optimal) Allocations for Binary Chores}
\author{
    Biaoshuai Tao\textsuperscript{\rm 1}\thanks{The authors are ordered alphabetically.}
    Xiaowei Wu\textsuperscript{\rm 2},
    Ziqi Yu\textsuperscript{\rm 1},
    Shengwei Zhou\textsuperscript{\rm 2}
}
\affiliations{
    \textsuperscript{\rm 1}Shanghai Jiaotong University,
    \{bstao, yzq.lll\}@sjtu.edu.cn\\
    \textsuperscript{\rm 2}University of Macau,
    \{xiaoweiwu, yc17423\}@um.edu.mo
}

\usepackage{bibentry}
\begin{document}

\maketitle

\begin{abstract}
We study the problem of allocating a group of indivisible chores among agents while each chore has a binary marginal.
We focus on the fairness criteria of envy-freeness up to any item (EFX) and investigate the existence of EFX allocations.
We show that when agents have additive binary cost functions, there exist EFX and Pareto-optimal (PO) allocations that can be computed in polynomial time.
To the best of our knowledge, this is the first setting of a general number of agents that admits EFX and PO allocations, before which EFX and PO allocations have only been shown to exist for three bivalued agents.
We further consider more general cost functions: cancelable and general monotone (both with binary marginal).
We show that EFX allocations exist and can be computed for binary cancelable chores, but EFX is incompatible with PO.
For general binary marginal functions, we propose an algorithm that computes (partial) envy-free (EF) allocations with at most $n-1$ unallocated items.
\end{abstract}

\section{Introduction}\label{sec:intro}
The {\em fair division} problem mainly focuses on allocating heterogeneous resources fairly to a group of agents.
It is a classic resource allocation problem that has been widely studied by mathematicians, economists, and computer scientists.
In this research direction, a popular question is to focus on allocating {\em indivisible} items.
An ideal fairness notion is {\em envy-freeness} (EF)~\cite{foley1967resource}, which requires that every agent values her own bundle not lower than any other bundle held by others.
However, EF allocations are not guaranteed to exist when allocating indivisible items, e.g. allocating one item to two agents.
% A trivial instance is given by one item and two agents, hence relaxation is necessary.
\citet{journals/teco/CaragiannisKMPS19} proposed a relaxed notion called \emph{envy-freeness up to any item} (EFX), which requires that for any pair of agents $i$ and $j$, $i$ does not envy $j$ after removing an arbitrary item from $j$'s bundle.
The existence of EFX allocations still remains open for general instances and is only known for some special cases, e.g., two agents with general valuations~\cite{plaut2020almost}, three agents with some special valuations~\cite{conf/sigecom/AkramiACGMM23,conf/sigecom/ChaudhuryGM20} and {\em bi-valued} instances~\cite{journals/tcs/AmanatidisBFHV21}.

In the traditional setting of fair division problems, items are assumed to have non-negative effects on agents.
Roughly speaking, giving an item to any agent will make her happier.
This kind of question is called fair division of \emph{goods}.
On the contrary, the complement setting is called fair division of \emph{chores}.
In this case, each item is of negative value to agents, and we use \emph{cost} functions to describe the preference of each agent.
Some definitions of fairness criteria (such as EFX) can be naturally extended to chores division \cite{journals/sigecom/AzizLMW22}.
However, most of the known results do not directly extend to the chores setting.
The minor differences between goods and chores settings nullify some existing techniques, such as Cycle-Elimination, when considering the allocation of chores~\cite{conf/ijcai/0002022}.
Another example can be given by the comparison of goods and chores on fair and efficient allocations.
As one of the influential results in the world of goods, \citet{journals/teco/CaragiannisKMPS19} showed that maximizing {\em Nash social welfare} guarantees the fairness of {\em envy-freeness up to one item} (EF1) (a fairness notion weaker than EFX) and the efficiency of Pareto-optimality (PO).
However, the existence of EF1 and PO allocations remains a major open problem in the context of chores, except for a limited class of bi-valued additive cost functions~\cite{conf/aaai/GargMQ22,conf/atal/EbadianP022}.

% Additional to results only regarding fairness, a variety of research focus on allocations that are both fair and efficient.
% As one of the most well-studied efficiency properties, Pareto-optimality is usually considered with the fairness notion of {\em envy-freeness up to one item} (EF1)~\cite{conf/sigecom/LiptonMMS04}, which is a property weaker than EFX.
% For the allocation of goods, Caragiannis et al.~\cite{journals/teco/CaragiannisKMPS19} showed that EF1 and PO allocations are guaranteed to exist by maximizing {\em Nash social welfare}.
% In contrast to the allocation of goods, the existence of EF1 and PO remains a major open problem in the context of chores.
% Except for a limited class of bi-valued costs~\cite{conf/aaai/GargMQ22,conf/atal/EbadianP022}, the existence of EF1 and PO allocations for chores is less known.
% When regarding stronger fairness of EFX, the incompatibility of EFX and PO has been shown in 

Major of our results focus on the settings when the agents' costs for chores have \emph{binary marginals}.
To be exact, for a binary cost function $c$, $c(S\cup\{g\})-c(S)\in\{0,1\}$, where $S$ is any set of chores and $g$ is a single chore.
Binary marginals have received much attention in the literature for both goods~\cite{babaioff2021fair,conf/wine/BarmanV21,journals/teco/KurokawaPS18} and chores~\cite{conf/atal/BarmanNV23}, due to various real-world applications such as bilateral matching~\cite{bogomolnaia2004random} and housing allocations~\cite{conf/sagt/BenabbouCIZ20}.
The binary marginal gives us extra properties to develop new techniques.

\subsection{Our results}
% \bt{I weakly prefer the first figure, although I am also OK with the second one. Is it possible to put the theorem reference (and the paper reference for the supermodular result) next to each result (probably we need to make the figure larger)? The current version may make the reader think that the results for the supermodular case are also ours.}
In this paper, we study the fair (and efficient) allocations of indivisible chores when each chore has a binary marginal cost.
We consider the fairness notion of envy-freeness up to any item (EFX) and investigate the existence of EFX allocations for different cost functions.
For additive cost functions with binary marginal, we show that EFX and Pareto-optimal (PO) allocations exist and can be computed in polynomial time.
Prior to our result, it has only been shown that EFX and PO allocations exist for restricted instances with three bivalued agents~\cite{journals/corr/abs-2212-02440}.
This is the first setting, even for additive cost functions, that admits EFX and PO allocations for a general number of agents.

\smallskip
\noindent
{\bf Result 1} (Theorem~\ref{thm:additive}) {\bf .} There exists a polynomial time algorithm that computes EFX and PO allocations for additive cost functions with binary marginals.

We further consider the {\em cancelable} cost functions (with binary marginals) that are first introduced by \citet{conf/aaai/BergerCFF22}.
Cancelable cost functions generalize several cost functions including {\em additive}, {\em budget additive}, and more.
We show that EFX is no longer compatible with PO for the more general cancelable cost functions.
We propose a polynomial-time algorithm that computes EFX allocations for binary cancelable cost functions.
To the best of our knowledge, this is the first non-trivial result regarding the existence and computation of EFX allocations for non-identical cost functions that are beyond additive.

\smallskip
\noindent
{\bf Result 2} (Theorem~\ref{thm:main_cancelable}) {\bf .} There exists a polynomial time algorithm that computes EFX allocations for cancelable cost functions with binary marginals.

We further consider submodular and general cost functions with binary marginals.
We show that under the binary marginal property, the set of submodular cost functions is a superset of the set of cancelable cost functions.
We propose a polynomial-time algorithm that computes $2$-approximately EFX allocations for binary submodular cost functions, and (partial) EF allocations with at most $n-1$ unallocated items for general cost functions with binary marginals.

\smallskip
\noindent
{\bf Result 3} (Theorem~\ref{thm:general}) {\bf .} There exists a polynomial time algorithm that computes a partial allocation that is envy-free and leaves at most $n-1$ items unallocated, for general cost functions with binary marginals.

\smallskip
\noindent
{\bf Result 4} (Theorem~\ref{thm:submodular}, Corollary~\ref{cor:submodular}) {\bf .} There exists a polynomial time algorithm that computes $2$-approximately EFX allocations for submodular cost functions with binary marginals.

We summarize our results in Figure~\ref{fig:relationship} which also presents the relationship between different cost functions with binary marginals.

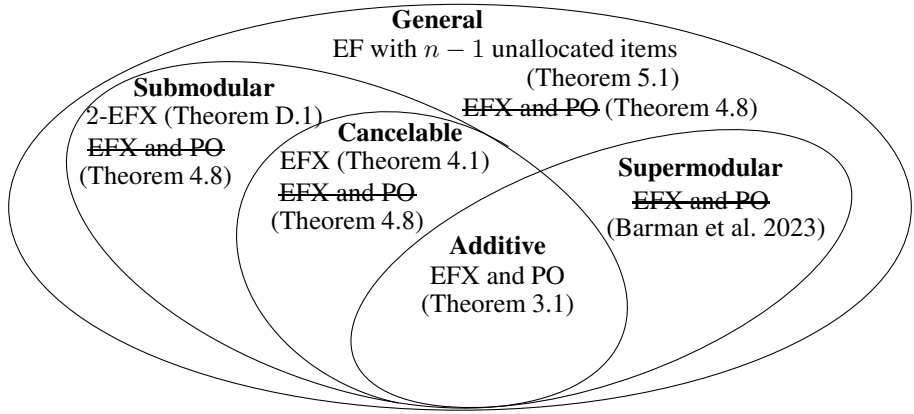
\begin{figure*}[htbp]
    \begin{center}
    \begin{tikzpicture}
        \draw [rotate around ={20:(0, 0)}] (2, 0) ellipse (3.5 and 1.5);
        \draw [rotate around ={340:(0, 0)}] (-1.8,0.55) ellipse (3.9 and 2);
        \node at (0.5, 1) {\textbf{Additive}};
        \node at (0.5, 0.6) {EFX and PO};
        \node at (0.5, 0.2) {(Theorem~\ref{thm:additive})};
        \draw [rotate around ={340:(0.5,2.2)}] (0.6,2.35) arc (75:280:2.45 and 1.9);
        \node at (-0.8, 2.5) {\textbf{Cancelable}};
        \node at (-1, 2.1) {EFX (Theorem~\ref{thm:main_cancelable})};
        \node at (-1.5, 1.7) {\st{EFX and PO}};
        \node at (-1.5, 1.3) {(Theorem~\ref{thm:cancelable-po})};
        \node at (-3.4, 3.1) {\textbf{Submodular}};
        \node at (-3.4, 2.7) {$2$-EFX (Theorem~\ref{thm:submodular})};  
        \node at (-4.05, 2.3) {\st{EFX and PO}};
        \node at (-4.05, 1.9) {(Theorem~\ref{thm:cancelable-po})}; 
        \node at (3.2, 2) {\textbf{Supermodular}};
        \node at (3.2, 1.6) {\st{EFX and PO}};
        \node at (3.4, 1.2) {(Barman et al. \citeyear{conf/atal/BarmanNV23})};
        \draw (0, 1.5) ellipse (6 and 2.7);
        \node at (-0.3, 4) {\textbf{General}};
        \node at (0.6, 3.6) {EF with $n-1$ unallocated items};
        \node at (1.9, 3.2) {(Theorem~\ref{thm:general})};
        \node at (2, 2.8) {\st{EFX and PO} (Theorem~\ref{thm:cancelable-po})};
    \end{tikzpicture}
    \end{center}
    \vspace{-15pt}
    \caption{Illustration of our results and the relationship between different set functions with binary marginals}
    \label{fig:relationship}
\end{figure*}

\subsection{Other Related Work}
The setting of binary marginal cases is widely considered in recent years.
In the literature of goods, \citet{babaioff2021fair} and \citet{conf/wine/0002PP020} considered the class of binary submodular (or matroid-rank) valuations and showed that EFX and PO allocations always exist.
They designed a mechanism that maximizes \emph{Nash social welfare} which generates EFX and PO properties under the matroid-rank valuations.
\citet{conf/atal/ViswanathanZ23} improved this by proposing the ``Yankee-Swap algorithm'' that runs significantly faster.
However, maximizing Nash social welfare might fail to guarantee EFX allocations when it comes to more general valuations, which makes it impossible to extend Babaioff et al.'s and Halpern et al.'s mechanisms directly.
Very recently, \citet{journals/corr/abs-2308-05503} developed a polynomial-time algorithm, based on cycle-elimination techniques, that computes EFX allocations for general valuations with binary marginals.
However, these results can not be easily extended to that of chores.
In the setting of chores division, the cycle-elimination process will sometimes break the EFX property, which will not happen in the world of goods division.
% This makes the binary chores division problem a challenging one.
When it comes to binary chores division, \citet{conf/atal/BarmanNV23} considered binary supermodular costs and showed that: 1) EF1 and PO allocations exist and can be computed in polynomial time; 2) EFX and PO are incompatible, under the binary supermodular valuations.
However, the existence of EFX allocations for non-identical cost functions and the compatibility of EFX and PO for other binary functions, e.g., {\em additive} and {\em submodular}, are still open.

Though the existence of EFX allocations remains open, much partial progress has been made over the past years.
A line of the literature focuses on exploring under which restriction EFX allocations exist. 
For the allocation of goods, \citet{plaut2020almost} showed that EFX allocations exist when all agents have identical valuations.
On the restriction of the number of agents, EFX allocations have been shown to exist for two agents with general valuations~\cite{plaut2020almost} and three agents with additive valuations~\cite{conf/sigecom/ChaudhuryGM20}.
The latter result was recently extended to more general valuations by \citet{conf/sigecom/AkramiACGMM23}.
In the context of chores, it has been shown that EFX allocations exist for some special cases, e.g., agents have monotone identical functions~\cite{journals/corr/abs-2006-04428}, identical ordering (IDO) instances~\cite{conf/www/0037L022}, three agents with bivalued additive functions~\cite{conf/ijcai/0002022}, two types of chores~\cite{conf/atal/0001LRS23}.
Another branch of the literature is to explore the relaxations of EFX, e.g., approximately EFX allocations~\cite{plaut2020almost,conf/atal/Chan0LW19,journals/tcs/AmanatidisMN20,conf/ijcai/0002022} and partial EFX allocations with reasonable guarantees~\cite{journals/siamcomp/ChaudhuryKMS21,conf/aaai/BergerCFF22,conf/ec/CaragiannisGH19,conf/sigecom/ChaudhuryGMMM21,journals/corr/abs-2212-09482}.
For a more comprehensive overview of the fair division problem, we refer to the recent surveys by \citet{journals/ai/AmanatidisABFLMVW23} and ~\citet{journals/sigecom/AzizLMW22}.

%Besides EFX, several other fairness notions such as {\em maximin share} (MMS)~\cite{conf/bqgt/Budish10} and {\em proportionality up to any item} (PROPX)~\cite{journals/annurev/Herve19} have also been widely studied in the allocation of chores.
%exist~\cite{conf/wine/FeigeST21}, a list of works~\cite{conf/aaai/AzizRSW17,aziz2022approximate,conf/sigecom/HuangL21} studied approximation of MMS and lead to the state-of-the-art approximation ratio of $13/11$~\cite{conf/sigecom/HuangS23}.
%Li et al. considered the fairness of PROPX and showed the existence and computation of PROPX allocations for chores~\cite{conf/www/0037L022}.
%For a more comprehensive overview of the fair division problem, we refer to the recent surveys by Amanatidis et al.~\cite{journals/ai/AmanatidisABFLMVW23} and Aziz et al.~\cite{journals/sigecom/AzizLMW22}.

\section{Preliminaries}\label{sec:preli}

We consider how to fairly allocate a set of $m$ indivisible chores/items $M$ to a group of $n$ agents $N$.
We call a subset of items, e.g., $S \subseteq M$, a \emph{bundle}.
For ease of notation we use $X+e$ and $X-e$ to denote $X\cup \{e\}$ and $X \setminus \{e\}$, respectively, for any $X\subseteq M$ and $e\in M$.
% For any integer $i\geq 1$ we use $[i]$ to denote $\{1,2,\ldots,i\}$.
% 
A complete allocation $\bX = (X_1,\ldots,X_n)$ is an $n$-partition of the items $M$ such that $X_i \cap X_j = \emptyset$ for all $i \neq j$ and $\cup_{i\in N} X_i = M$, where agent $i$ receives bundle $X_i$.
When $\cup_{i\in N} X_i \subsetneq M$, we call it a partial allocation.
Given an instance $I = (N, M, \bc)$ where $\bc=(c_1,\ldots,c_n)$ is the set of \emph{cost functions} (to be defined immediately), our goal is to find an allocation $\bX$ that is \emph{fair} to all agents.

\paragraph{Cost Functions.}
Each agent $i\in N$ has a cost function $c_i: 2^M \to \bR^+\cup \{0\}$ that assigns a cost to every bundle of items.
More specifically, for any subset of items $S\subseteq M$, the cost of $S$ to agent $i$ is denoted as $c_i(S)$.
For convenience we use $c_i(e)$ to denote $c_i(\{e\})$, the cost of agent $i\in N$ on item $e\in M$.
% For convenience we use $c_i(e)$ to denote $c_i(\{e\})$, the value of agent $i\in N$ on item $e\in M$, and thus $v_i(X) = \sum_{e\in X} c_i(e)$ for all $X\subseteq M$.
We use $\bc = (c_1,\ldots, c_n)$ to denote the cost functions of agents.
For any subset of items $S\subseteq M$, any item $e\in M$, we use $c_i(e\mid S)$ to denote the marginal cost of item $e$ to subset $S$, under agent $i$'s cost function, i.e., $c_i(S+e) - c_i(S)$.
Similarly, for $S,T\subseteq M$, we denote $c_i(T\mid S)=c_i(T\cup S)-c_i(S)$.
\begin{itemize}
    \item \emph{Binary Marginal.}
    We focus on the instances in which each item gives \emph{binary} marginal cost to any subset of items, that is, for any $i\in N, e\in M, S\subseteq M$, we have $c_i(e\mid S) \in \{0,1\}$.
    Note that, any binary marginal cost function is also \emph{monotone}, i.e., $c_i(S+e) \geq c_i(S)$ for any $i\in N, e\in M, S\subseteq M$.
    \item \emph{Additive Function.}
    A cost function $c_i(\cdot)$ is said to be additive if $v_i(S) = \sum_{e\in S} c_i(e)$ for any $S\subseteq M$.
    \item \emph{Cancelable Function.}
    As a generalization of additive cost functions, cancelable functions are first introduced by \citet{conf/aaai/BergerCFF22} in the fair allocation of goods.
    A cost function $c_i(\cdot)$ is said to be cancelable if for any two bundles
    $S, T \subseteq M$ and item $e\in M\setminus (S\cup T)$, we have 
    \begin{equation}
        c_i(S + e) > c_i(T+e) \Rightarrow c_i(S) > c_i(T).
    \end{equation}
    As ~\citet{conf/aaai/BergerCFF22} pointed out, cancelable functions describe many natural meaningful valuation functions in economics, including additive, unit-demand, and budget-additive cost functions.
    \item \emph{Submodular Function.}
    % Under the binary marginal constraints, we show that any cancelable function is also submodular.
    A cost function $c_i(\cdot)$ is submodular if and only if for any $S\subseteq T, e\in M\setminus T$, we have $c_i(e\mid S) \geq c_i(e\mid T)$.
\end{itemize}

In Appendix~\ref{sect:relationship}, we will show that, under set functions with binary marginals, the set of additive functions is a proper subset of the set of cancelable functions, and the set of cancelable functions is a proper subset of the set of submodular functions. Notice that the latter only holds for set functions with binary marginals.

\paragraph{Fairness.}
We first introduce the {\em envy-freeness} (EF) for the allocation of chores.
An allocation $\bX$ is EF if for any agents $i, j\in N$, $c_i(X_i) \leq c_i(X_j)$.
Since EF allocations are not guaranteed to exist for indivisible chores, we focus on a relaxation of EF: {\em envy-freeness up to any item} (EFX).
An allocation $\bX$ is EFX if for any agents $i, j\in N$, either $X_i = \emptyset$, or $c_i(X_i-e) \leq c_i(X_j)$ for any item $e\in X_i$.
For $\alpha \geq 1$, an allocation is $\alpha$-approximately envy-free, denoted by $\alpha$-EF, if $c_i(X_i)\leq\alpha\cdot c_i(X_j)$;
an allocation is $\alpha$-approximately EFX, denoted by $\alpha$-EFX, if either $X_i=\emptyset$ or $c_i(X_i-e)\leq \alpha\cdot c_i(X_j)$ for any $i,j\in N$ and any $e\in X_i$.

\paragraph{Efficiency.}
An allocation $\bX'$ \emph{Pareto-dominates} another allocation $\bX$ if $c_i(X'_i) \leq c_i(X_i)$ for all $i\in N$ and the inequality is strict for at least one agent.
An allocation $\bX$ is said to be {\em Pareto-optimal} (PO) if $\bX$ is not Pareto-dominated by any other allocation.
For any allocation $\bX = (X_1, \cdots, X_n)$, the social cost $\sc(\bX)$ is defined as the sum of the cost of each agent, i.e., $\sc(\bX)=\sum_{i\in N} c_i(X_i)$.
An allocation that minimizes the social cost is naturally Pareto-optimal since any Pareto-improvement strictly decreases the social cost.

\subsection{Relationships Between Set Functions}
In Appendix~\ref{sect:relationship}, we prove the following two theorems, which state that additive functions form a special case of cancelable functions, and cancelable functions form a special case of submodular functions, when we are considering functions with binary marginals. Notice that Theorem~\ref{thm:cancelable-submodular} only holds for functions with binary marginals.

\begin{restatable}{theorem}{thmAddCancelable}
    The set of binary additive set functions is a proper subset of the set of cancelable set functions with binary marginals.
\end{restatable}

%Next, we show that, by applying binary marginal property to cancelable functions, any cancelable function with binary marginals is also submodular.

\begin{restatable}{theorem}{thmCancelableSubmodular}\label{thm:cancelable-submodular}
    The set of cancelable set functions with binary marginals is a proper subset of the set of submodular set functions with binary marginals.
\end{restatable}

We also list some properties for cancelable and submodular functions.

\begin{restatable}{proposition}{propcancelable}\label{prop:cancelable}
    Let $\phi:M\to\mathbb{Z}_{\geq0}$ be a cancelable function.
    \begin{enumerate}
        \item For any $S,T\subseteq M$ and any $e\in M\setminus(S\cup T)$, if $\phi(S)=\phi(T)$, then $\phi(S+e)=\phi(T+e)$.
        \item For any $S,T\subseteq M$ and any $U\subseteq M$ with $U\cap (S\cup T)=\emptyset$, if $\phi(S)=\phi(T)$, then $\phi(S\cup U)=\phi(T\cup U)$.
    \end{enumerate}
\end{restatable}

%Finally, we prove some properties for submodular functions.
%By Theorem~\ref{thm:cancelable-submodular}, cancelable functions with binary marginals also satisfy these properties.

\begin{restatable}{proposition}{propsubmodular}\label{prop:submodular}
    Let $\phi:M\to\mathbb{Z}_{\geq0}$ be a submodular function with binary marginals.
    \begin{enumerate}
        \item For any $S,T\subseteq M$, $\phi(S)+\phi(T)\geq \phi(S\cup T)+\phi(S\cap T)$. In particular, $\phi(S)+\phi(T)\geq \phi(S\cup T)$.
        \item If $\phi(S)=1$ holds for some $S$ with $|S|\geq 2$, there exists $e\in S$ such that $\phi(S-e)=1$.
    \end{enumerate}
\end{restatable}

\section{Existence of EFX and PO Allocations}\label{sec:additive}
% In the additive setting, instances with binary marginals property can be generalized by the binary instances.
% We first give the definition of binary instances.

% \bt{I think we should mainly focus on the 0-1 binary special case here, as this is the main focus of our paper. The extension below can be briefly described after the result, with the details deferred to the appendix.}

% \begin{definition}[Binary Instances]
%     An instance is called \emph{binary} if for each agent $i\in N$ there exists a constant $a_i > 0$ such that for any item $e\in M$ we have $c_i(e) \in \{0,a_i\}$.
% \end{definition}

% Note that for each non-zero $a_i$ we can scale the cost functions so that $c_i(e) = \{0,1\}$ for all $i\in N$.
In this section, we explore the existence of EFX and PO allocations for additive cost functions.
We call an instance \emph{binary} if each agent has an additive valuation with binary marginals.
We propose an algorithm that computes EFX and PO allocations for binary instances.
This is the first class with a general number of agents, for which EFX and PO allocations exist and can be computed in polynomial time.
% On the other hand, we show that even extending the result to ternary instances is impossible, by providing a hardness in which EFX and PO are incompatible.
% 
In the world of allocation of goods, it has been shown that maximizing Nash social welfare (product of agents' utilities) leads to EFX and PO allocations for bivalued instances (which is a superset of binary instances)~\cite{journals/tcs/AmanatidisBFHV21,conf/sagt/GargM21}.
However, in the opposite of goods, minimizing Nash social welfare results in zero Nash social welfare which has no guarantee of fairness and efficiency.
Before our result, it has only been shown that EFX and PO allocations exist for three bivalued agents~\cite{journals/corr/abs-2212-02440}.

Given a binary instance, we divide the set of chores into $M^0$ and $M^+$, while $M^+$ includes those items that cost $1$ to all agents, i.e., $c_i(e) = 1$ for all $i\in N$.
\begin{align*}
    M^0 = \{e\in M: \exists i\in N, c_i(e) = 0\}, \\
    M^+ = \{e\in M: \forall i\in N, c_i(e) = 1\}.
\end{align*}
 In this section, we prove the following main result.
 
\begin{algorithm}[htbp]
    \caption{Finding an EFX and PO allocation for binary additive cost functions} \label{alg:additive}
    \KwIn{A binary instance $<M, N, \bc>$}
    initialize $X_i \gets \emptyset$ for all $i\in N$, $P \gets M$ \;

    \tcp{Phase 1: compute a partial allocation $\bX^0$ with $\sc(\bX^0) = 0$.}
    \For{each item $e\in M^0$}{
        pick an agent $i\in N$ such that $c_i(e) = 0$. break tie arbitrary\;
        update $X_i \gets X_i + e$\;
    }
    
    \tcp{Phase 2: allocate items in $M^+$.}
    
    initialize $P \gets M^+$ \;
    \While{$P \neq \emptyset$}{
        let $i^* \gets \argmin_{i\in N} \{c_i(X_i)\}$, break tie arbitrary\;
        update $X_{i^*} \gets X_{i^*} + e$, $P\gets P - e$\;
        \If{there exists an agent $j\neq i^*$ such that $i^*$ is not EFX towards $j$}{
            update $X_{i^*} \gets X_{i^*}-e, X_j \gets X_j + e$\;
            \While{there exists an item $e'\in X_j$ such that $c_{i^*} (e') = 0$}{
                update $X_{i^*} \gets X_{i^*} + e', X_j \gets X_j - e'$\;
            }
        }
    }
    \KwOut{$\bX = \{X_1, \cdots, X_n\}$.}
\end{algorithm}
% \vspace{-5pt}

\begin{theorem}\label{thm:additive}
    There exists an algorithm that computes EFX and PO allocations for any given binary instance in polynomial time.
\end{theorem}
Our algorithm starts from a partial allocation $\bX^0$ in which every agent only receives items that have zero cost to her.
In the second phase, we allocate each item $e\in M^+$ to one agent straightly if it maintains EFX property after allocating $e$, or we reallocate some items until an item $e\in M^+$ can be allocated without breaking EFX.
The steps of the full algorithm are summarized in Algorithm~\ref{alg:additive}.
\begin{lemma}\label{lemma:minmized-sc}
    Algorithm~\ref{alg:additive} returns an allocation $\bX$ with minimum social cost.
\end{lemma}
    \begin{proof}
        We prove the lemma by showing that $\sc(\bX) = |M^+|$, while no complete allocation has a strictly less social cost.
        In Phase 1 we compute a partial allocation $X^0$ with zero social cost and remaining all items in $M^+$ unallocated.
        Then in Phase 2 we allocate items in $M^+$ and reallocate some items, while each allocation of item $e\in M^+$ increases the total social cost by $1$.
        Note that item $e\in M$ would be reallocated to agent $i$ only under the case that $c_i(e) = 0$.
        Hence any reallocation would not change the total social cost of the allocation.
        In conclusion, for the returned allocation $\bX$ we have $\sc(\bX) = |M^+|$.
    \end{proof}

    \begin{lemma}\label{lemma:efx-additive}
        The allocation $\bX$ returned by Algorithm~\ref{alg:additive} is EFX to all agents.
    \end{lemma}
    \paragraph{Proof Sketch}
    % Due to the page limit, we provide a sketched proof here and refer the full version of proof to Appendix~\ref{app:additive-missing}.
    We show that the allocation $\bX$ is EFX by induction and treat the beginning of Phase 2 as the basic case.
    The statement holds for the basic case since $c_i(X_i) = 0$ for all $i\in N$.
    Then we assume it holds at the beginning of some round $t$ of Phase 2 and we show that it still holds at the end of round $t$.
    According to the algorithm, we only need to consider the case under the if-condition in line $9$.
    As for agent $i$, the envy between $i$ and other agents almost stays the same since agent $i$ only receives some items that cost $0$ to her.
    Then we consider the envy between $j$ and other agents.
    An important claim we used is that $c_i(X_i) = c_j(X_j)$ holds when the if-condition from Line $9$ to Line $12$ is executed, and $X_j\setminus S \subseteq M^+$ holds after the execution.
    The EFX property stays trivially for any agent rather than $j$.
    % All items in $X'_j$ cost $1$ to other agents and the EFX property stays trivially for any agent rather than $j$.
    The allocation is EFX for agent $j$ since after removing any item, she holds the minimum bundle over all agents.

Furthermore, it can be easily verified that the algorithm runs in polynomial time.
We further complete our result by exploring the non-existence of EFX and PO allocations in the additive setting.
We show that even extending our result to ternary instances\footnote{An instance is called ternary if all cost functions are additive and $c_i(e) \in \{0,1,2\}$ for all $i\in N, e\in M$.} is impossible (see Appendix~\ref{ssec:non-existence}).

\section{EFX Allocations Exist for Binary Cancelable Chores}
\label{sect:cancelable}
In this section, we show that, for cancelable cost functions with binary marginals, we can always find EFX allocations in polynomial time, but there exist instances where all EFX allocations are not Pareto-optimal.
%We will also show that our algorithm can be extended to show the existence of $\frac12$-EFX allocations for submodular cost functions.

\begin{theorem}\label{thm:main_cancelable}
    For cancelable cost functions with binary marginals, EFX allocations always exist and can be computed in polynomial time.
\end{theorem}

\begingroup
\removelatexerror% Nullify \@latex@error
\begin{algorithm}[H]
    \caption{Finding an EFX allocation for cancelable cost functions with binary marginals} \label{alg:cancelable}
		\KwIn{A binary instance $(M, N, \bc)$}
		\tcp{Phase 1: allocate items with marginal cost~1 from all agents' perspective}
		initialize $(A_1,\ldots,A_n)$ with $A_i=\emptyset$ for all $i$\;
        \While{there exist $n$ items such that each item $e$ satisfies $c_i(e\mid A_i)=1$ for all $i$}{
            add those $n$ items to $(A_1,\ldots,A_n)$ such that each $A_i$ is added one item\;
        }
        let $w=|A_1|=\cdots=|A_n|$\;
		\tcp{Phase 2: allocate the remaining items}
        initialize $(B_1,\ldots,B_n)$ where $B_i=\emptyset$ for each $i$\;
        consider the cost function $\bd=(d_1,\ldots,d_n)$ with $d_i(S)=c_i(S\mid A_i)$\;
        let $M_1=\{e\in M: d_i(e)=1\mbox{ for all }i\}$\; \tcp{We have $|M_1|<n$, for otherwise Phase~1 should not have been terminated.}
		for each $i=1,\ldots,|M_1|$, let $B_i$ be the set containing one (arbitrary) item of $M_1$\;
        \While{there is an unallocated item $e$}{
        \tcp{We maintain that $(B_1,\ldots,B_n)$ is EFX with respect to $\bd$ and $d_i(B_i)\leq 1$ for any $i$.}
        \eIf{$d_i(e\mid B_i)=0$ for some agent $i$ and $(B_1,\ldots,B_n)$ is EFX after updating $B_i\leftarrow B_i+e$}{
            update $B_i\leftarrow B_i+e$\;
        }{
            \tcp{If we reach here, there must exist $i$ and $j$, with the possibility $i=j$, such that $d_i(B_j)=0$ (Proposition~\ref{prop:existzero}).}
            \eIf{there exists $i\in N$ with $d_i(B_i)=0$}{
            if there exists $j\neq i$ with $d_i(B_j)=0$, update $B_i\leftarrow B_i\cup B_j$ and $B_j\leftarrow\{e\}$;
            otherwise, update $B_i\leftarrow B_i+e$\;
            }{
            find $j$ such that $d_i(B_j)=0$, and swap $B_i$ and $B_j$\;
            }
        }}
		\KwOut{$\bX = (X_1, \ldots, X_n)$ where $X_i=A_i\cup B_i$ for each $i$.}
\end{algorithm} 
\endgroup

Our algorithm is presented in Algorithm~\ref{alg:cancelable}.
The algorithm consists of two phases.
In the first phase, we iteratively allocate $n$ items, where each agent is believed to have a marginal cost of $1$, such that each agent is allocated exactly one item.
Let $w$ be the number of items allocated to each agent, and we have allocated $wn$ items.
After Phase~1, the number of items with marginal costs $1$ to all agents is then less than $n$.
Let $M_1$ be the set of these items.
Let $(A_1,\ldots,A_n)$ be the allocation of Phase~1.
After the first phase, each agent believes that her own bundle costs $w$ and each other agent's bundle also costs $w$, as the lemma below states.

\begin{lemma}\label{lem:phase1}
    After Phase~1, we have $c_i(A_j)=w$ for each pair of $i$ and $j$, with the possibility $i=j$.
\end{lemma}
\begin{proof}
    This can be proved by straightforward induction. See Appendix~\ref{sect:missingproofcancelable} for the full proof.
\end{proof}
% \begin{proof}
%     We prove by induction that $c_i(A_j)=t$ for each pair of $i$ and $j$ after $t$ while-loop iterations.
%     The base step for $t=0$ is trivial.
%     Suppose the claim holds for $t$.
%     We will show it holds for $t+1$.
%     Let $(A_1,\ldots,A_n)$ be the allocation before the $(t+1)$-th iteration.
%     By induction hypothesis, $c_i(A_j)=t$ for every pair of $i$ and $j$.
%     For each item $e$ allocated in the $(t+1)$-th iteration, we have $c_i(e\mid A_i)=1$ for each agent $i$.
%     By Property~1 in Proposition~\ref{prop:cancelable}, this also implies $c_i(e\mid A_j)=1$ for every other agent $j$.
%     Thus, $c_i(A_j)=t+1$ for every pair of $i$ and $j$ after the $(t+1)$-th iteration.
% \end{proof}

In Phase 2, we compute an allocation $(B_1,\ldots,B_n)$ of the remaining items.
We update the cost function $c_i$ such that, for each unallocated set $S$ of items, the cost of $S$ is given by $c_i(S\mid A_i)$.
Let $\bd$ be the set of the updated cost functions.
We will show that the allocation $(B_1,\ldots,B_n)$ output in Phase~2 satisfies the following two key properties:
\begin{itemize}
    \item $d_i(B_i)\leq 1$ for each agent $i$.
    \item $(B_1,\ldots,B_n)$ is EFX with respect to $\bd$.
\end{itemize}

Lastly, the final allocation $\bX=(X_1, \ldots, X_n)$ is given by $X_i=A_i\cup B_i$ for each $i\in N$.
We will show that $\bX$ is EFX with respect to $\bc$.
The cancelable property plays an important role in guaranteeing EFX property when combining the allocations in the two phases.
Specifically, the cancelable property and Lemma~\ref{lem:phase1} ensure a good interpolation between the two cost function sets $\bd$ and $\bc$.

Before proving the key properties of Phase~2 allocation, we first state the following proposition whose proof is straightforward.

\begin{proposition}\label{prop:d_i}
    Each $d_i$ is a cancelable and submodular set function with binary marginals.
\end{proposition}
% \begin{proof}
%     It is straightforward to check $d_i$ is cancelable by definition and that $c_i$ is cancelable.
%     It is also trivial that $d_i$ has binary marginals.
%     Submodularity of $d_i$ follows from Theorem~\ref{thm:cancelable-submodular}.
% \end{proof}

Now we are ready to show our main lemma for Phase~2.
It states that the algorithm will always be terminated and the allocation output satisfies our two key properties.

\begin{lemma}\label{lem:phase2}
    The while-loop in Phase~2 is executed for at most $2m$ iterations, and the output allocation $(B_1,\ldots,B_n)$ satisfies that
    \begin{enumerate}
        \item $d_i(B_j)\leq 1$ for each agent $i$, and
        \item $(B_1,\ldots,B_n)$ is EFX with respect to $\bd$.
    \end{enumerate}
\end{lemma}

We first show that Property~1 and 2 above are always satisfied.

\begin{proposition}\label{prop:keyproperties}
    After any number of while-loop iterations of Phase~2, the two properties 1 and 2 in Lemma~\ref{lem:phase2} are satisfied.
\end{proposition}
\begin{proof}
    Before entering the while-loop, exactly $|M_1|$ bundles contain one item with cost $1$ from all agent's perspectives, and the remaining $n-|M_1|$ bundles are empty. The properties clearly hold.
    Next, suppose the properties hold before a while-loop execution, we will show that they continue to hold after.
    We will check all the updates, at Line~11, 14, and 16, do not invalidate the properties.
    
    The update at Line~11 adds an item with marginal cost $0$ to the bundle $B_i$, so Property~1 continues to hold.
    It will not invalidate Property~2 by the if-condition.
    
    For the update at Line~14, we have $d_i(B_i\cup B_j)\leq d_i(B_i)+d_i(B_j)=0$ by submodularity of $d_i$ (Proposition~\ref{prop:d_i}), $d_i(B_i+e)\leq 1$, and $d_j(e)\leq 1$. In all cases, Property~1 continues to hold.
    To check Property~2, if there does not exist $j\neq i$ with $d_i(B_j)=0$, then agent $i$ receives $B_j+e$ and the allocation for the remaining agents is unchanged.
    In this case, $d_i(B_j)\geq1$, and agent $i$ does not envy any other agent by receiving $B_i+e$, which has cost at most $1$. The EFX condition between any other agent and $i$ still holds as $i$ receives a superset of $B_i$.
    If $d_i(B_j)=0$ for some $j$, agent $i$ receives $B_i\cup B_j$ and agent $j$'s bundle is updated to the singleton $\{e\}$.
    We have $d_i(B_i\cup B_j)\leq d_i(B_i)+d_i(B_j)=0$, so $i$ does not envy any other agent. On the other hand, $j$ receives only a single item $e$, she will no longer envy any other agent if this item were removed.
    Consider any other agent $k$ that is neither $i$ nor $j$.
    The EFX condition between $k$ and $i$ continues to hold, as $i$ now receives a superset of $B_i$.
    As for the EFX condition between $k$ and $j$, we discuss two cases.
    If $d_k(e)=1$, then agent $k$ does not envy agent $j$ as we have $d_k(B_k)\leq 1$ by Property~1.
    If $d_k(e)=0$, we have $d_k(e\mid B_k)\leq d_k(e)=0$ by submodularity.
    Since the if-condition at Line~10 fails (otherwise, the ``else'' part will not be executed), adding $e$ to $B_k$ will break the EFX condition between $k$ and some other agent $\ell$.
    By Property~1, it must be that $d_k(B_k)=1$ and $d_k(B_\ell)=0$.
    Moreover, $d_k(B_k-f)=0$ for any $f\in B_k$ (in fact, it must be that $B_k=\{f\}$ by Property~2 in Proposition~\ref{prop:submodular}).
    In this case, the EFX condition between $k$ and any other agent still holds.
    
    For the update at Line~16, if we have ever reached here, it must be that $d_i(B_i)=1$ and $d_i(B_j)=0$.
    Moreover, we must have $|B_i|=1$. Otherwise, by Property~2 of Proposition~\ref{prop:submodular}, the EFX property between $i$ and $j$ will be violated.
    After the update at Line~16, we have $d_i(B_i)=0$ and $d_j(B_j)\leq1$, and the latter holds because $B_j$, which is previously $B_i$, contains only one item.
    This proves Property~1.
    To check Property~2, $i$ does not envy any other agents as $d_i(B_i)$ is now $0$. The EFX condition between $j$ and any other agents holds as $j$'s bundle contains only one item now. The EFX condition between any other agent $k$ and $i$ or $j$ still holds as we have only swapped the bundles of $i$ and $j$.
\end{proof}

To show the algorithm terminates, we prove the following observation which is also stated between Line~12 and Line~13 of the algorithm.

\begin{proposition}\label{prop:existzero}
    If the ``else'' part from Line~12 to Line~16 is executed, there must exist $i$ and $j$, with the possibility $i=j$, such that $d_i(B_j)=0$.
\end{proposition}
\begin{proof}
    Suppose $d_i(B_j)\geq1$ for every pair of $i$ and $j$.
    We will show that the if-condition at Line~10 is always satisfied.
    Proposition~\ref{prop:keyproperties} indicates that Property~1 in Lemma~\ref{lem:phase2} holds before the if-condition at Line~10.
    Therefore, the current allocation $(B_1,\ldots,B_n)$ is envy-free.
    Adding an item with a zero marginal cost to an agent will not destroy the envy-free property.
    On the other hand, for each $e\notin M_1$, there exists an agent $i$ with $d_i(e)=0$, and, by submodularity, $d_i(e\mid B_i)=0$.
    Since all the items in $M_1$ have been allocated before the while-loop and the algorithm never returns an item back to the pool of unallocated items, the if-condition at Line~10 will be satisfied.
\end{proof}

Now we are ready to prove Lemma~\ref{lem:phase2}.

\begin{proof}[Proof of Lemma~\ref{lem:phase2}]
    Proposition~\ref{prop:keyproperties} ensures that the two properties always hold.
    It remains to show that the while-loop will be executed for at most $2m$ iterations.
    It suffices to show that at least one item is allocated in every two iterations.
    The only case where no item is allocated is when Line~16 is executed.
    In this case, Proposition~\ref{prop:existzero} ensures the existence of $j$.
    After the swapping, we have $d_i(B_i)=0$.
    Therefore, in the next iteration, either Line~11 or Line~14 will be executed, in which case an item is allocated.
\end{proof}

\begin{remark}\label{remark:phase2}
    For the proof of Lemma~\ref{lem:phase2}, we have only exploited the submodularity of $d_i$, not the cancelability.
    In particular, if no item is allocated in Phase~1 and the algorithm starts with Phase~2, we have $\bd=\bc$.
    The algorithm computes an EFX allocation even if each $c_i$ is only known to be submodular (while not necessarily cancelable).
\end{remark}

Finally, we combine Lemma~\ref{lem:phase1} and Lemma~\ref{lem:phase2} to conclude Theorem~\ref{thm:main_cancelable}.

\begin{proof}[Proof of Theorem~\ref{thm:main_cancelable}]
    We first show that the allocation output by Algorithm~\ref{alg:cancelable} is EFX.
    We check the EFX condition for an arbitrary pair of agents $i$ and $j$.
    By Lemma~\ref{lem:phase2}, $(B_1,\ldots,B_n)$ is EFX with respect to $\bd$, and we consider two cases: 1) $d_i(B_i)\leq d_i(B_j)$ and 2) $d_i(B_i)>d_i(B_j)$.

    For the first case, we have
    \begin{align*}
        c_i(X_i)&=c_i(A_i)+d_i(B_i)\leq c_i(A_i)+d_i(B_j)\tag{case assumption}\\
        &=c_i(A_i)+(c_i(A_i\cup B_j)-c_i(A_i))=c_i(A_i\cup B_j)\tag{definition of $d_i$}\\
        &=c_i(A_j\cup B_j) =c_i(X_j),\tag{Property~2 of Proposition~\ref{prop:cancelable} and Lemma~\ref{lem:phase1}}
    \end{align*}
    which indicates that $i$ does not envy $j$.

    For the second case, it must be that $|B_i|=1$.
    To see this, Property~1 of Lemma~\ref{lem:phase2} indicates that $d_i(B_i)=1$ and $d_i(B_j)=0$.
    Property~2 of Proposition~\ref{prop:submodular} indicates that $|B_i|\geq 2$ will destroy the EFX property (Property~2 of Lemma~\ref{lem:phase2}) between $i$ and $j$.
    Since we have seen $|B_i|=1$, $|X_i|=w+1$.
    Therefore, agent $i$'s cost will be at most $w$ after removing any item from $X_i$.
    On the other hand, by Lemma~\ref{lem:phase1} and the monotonicity, $c_i(X_j)\geq c_i(A_j)=w$.
    Thus, the EFX condition between $i$ and $j$ is satisfied.

    Lastly, checking that the algorithm runs in polynomial time is straightforward.
    % The first phase of the algorithm can be completed by $O(nm)$ time since verifying the marginal cost for each item costs at most $O(n)$.~\footnote{under the reasonable assumption that the cost of any bundle can be returned by an oracle in $O(1)$.}
    % In the second phase, the while-loop will be executed for at most $2m$ iterations, while each iteration can be done by $O(n)$ time.
    % In conclusion, the algorithm runs in $O(nm)$ time.
\end{proof}

We have seen that EFX allocations always exist for binary cancelable chores.
However, unlike the case with additive cost functions, EFX is no longer compatible with PO (see Appendix~\ref{sect:missingproofcancelable}).
This is in sharp contrast to the setting with goods.
For allocating goods, we know that EFX is compatible with PO even for the more general submodular binary cost functions: \citet{babaioff2021fair} show that the allocation maximizing the Nash social welfare (product of agents' utilities) is EFX.

\begin{theorem}\label{thm:cancelable-po}
    For any number of agents $n\geq 2$, there exist instances with binary cancelable cost functions where all EFX allocations are not Pareto-optimal.
\end{theorem}
% \begin{proof}
%     See Appendix~\ref{sect:missingproofcancelable}.
% \end{proof}
% \begin{proof}
%     Consider $n$ agents with $5n$ items where each agent's cost function is given by Equation~(\ref{eqn:cancelable}).
%     It is easy to check that the only EFX allocation gives exactly $5$ items to each agent.
%     However, this allocation is Pareto-dominated by the allocation that allocates all items to a single agent.
% \end{proof}

\section{General Cost Functions with Binary Marginals}
For general cost functions with binary marginals, we show that there exists a partial allocation that leaves at most $n-1$ items unallocated and is envy-free.
We will also show that the algorithm can be used to find a complete $2$-EFX allocation for submodular binary cost functions (we defer this part to Appendix~\ref{sec:submodular} due to the page limit).
We remark that we do not know if complete EFX allocations always exist, even for the more special case with submodular binary cost functions.

% \subsection{General Cost Functions}
\begin{theorem}\label{thm:general}
    For cost functions with binary marginals, there exists a partial allocation that is envy-free and leaves at most $n-1$ items unallocated. Moreover, such an allocation can be computed in polynomial time.
\end{theorem}

The algorithm makes use of the \emph{envy graph}, a technique that has been widely used in the fair division literature.
%However, we remark that it is mostly used for allocating goods, and inappropriate usage of it for allocating chores can be problematic.

\begin{definition}
    Given a partial allocation $(X_1,\ldots,X_n)$ and the cost function profile $\bc$, the \emph{envy graph} $G=(V=[n],E)$ is a directed graph where each vertex in $V$ represents an agent and $(i,j)\in E$ if and only if $v_i(A_i)=v_i(A_j)$.
\end{definition}

As a remark, unlike it is in most of the previous literature where an edge represents ``envy'', we always maintain an envy-free allocation and an edge in our envy graph represents equality, or, ``about to envy''.

\begin{algorithm}[h]
\caption{Finding a partial EF allocation for cost functions with binary marginals} \label{alg:general}
\KwIn{A binary instance $(M, N, \bc)$}
initialize $(X_1,\ldots,X_n)$ with $n$ empty bundles, and initialize the envy graph $G=(V,E)$\;
\While{there exist unallocated items}{
update the envy graph $G=(V,E)$\;
\If{there exist an unallocated item $e$ and $i$ with $c_i(e\mid X_i)=0$}{
    update $X_i\leftarrow X_i+e$\;
    \textbf{continue}\;
}
\If{there exist an unallocated item $e$ and an edge $(i,j)\in E$ such that $(i,j)$ is on a directed cycle $C$ and $c_i(e\mid X_j)=0$}{
    rotate the bundles on the cycle: for each edge $(u,v)\in C$, allocate $X_v$ to agent $u$\;\tcp{In particular, $i$ receives $X_j$.}
    add $e$ to agent $i$'s bundle\;
    \textbf{continue}\;
}
\tcp{Handling the case where the above two kinds of updates fail.}
find a tail strongly connected component $S$ in $G$\;\tcp{A tail strongly connected component $S$ has no outgoing edge from $S$; in particular, a sink is a tail strongly connected component.}
%\tcp{a sink can also be considered as a tail strongly connected component.}
\eIf{there are at least $|S|$ unallocated items}{
allocate each agent $S$ an arbitrary unallocated item\;
}{
\textbf{break}\;
}
}
\KwOut{$\bX = (X_1, \ldots, X_n)$}
\end{algorithm} 

\paragraph{Algorithm description and proof of Theorem~\ref{thm:general}.}
Our algorithm is presented in Algorithm~\ref{alg:general}.
It initializes the allocation with $n$ empty bundles, and the initial envy graph is a complete graph.
At each iteration, it attempts to allocate one or more items by using one of the three update rules:
\begin{enumerate}
    \item if there is an item with marginal cost $0$ to some agent, allocate it;
    \item if an edge $(i,j)$ is on a cycle $C$ and $i$ thinks adding some item $e$ to $j$'s bundle does not increase its cost, rotate the bundles on the cycle so that $i$ receive $j$'s bundle, and add $e$ to $i$'s bundle (which is previously $j$'s) which does not increase the cost for $i$;
    \item if the first two update rules do not apply, find a tail strongly connected component $S$ (a strongly connected component with no outgoing edges) and allocate each agent in $S$ an arbitrary item; if the number of items is insufficient for this, the algorithm is terminated with the partial allocation outputted.
\end{enumerate}

We will show that the EF property is satisfied throughout the algorithm.
The initial allocation is clearly envy-free.
It suffices to show that, given a partial allocation, any of the three update rules does not destroy envy-freeness.

%The first update rule obviously does not invalidate envy-freeness.

%For the second update rule, rotating the cycle does not invalidate envy-freeness: the constituents of the $n$ bundles never change, and the cost for each agent on the cycle does not change by our definition of edges in $G$.
%After the rotation, adding $e$ to $i$'s bundle does not increase $i$'s cost, and the allocation remains envy-free.

It is straightforward to check the first and the second update rules do not invalidate envy-freeness.

For the third update rule, first notice that the envy-freeness between an agent $i$ in $S$ and an agent $k$ in $V\setminus S$ is preserved.
Agent $k$ does not envy agent $i$ as agent $i$'s cost can only be increased from $k$'s perspective.
Since $S$ is a tail strongly connected component, $(i,k)\notin E$ before the update, so $v_i(X_i)\leq v_i(X_k)-1$.
Adding an item to agent $i$, which increases $v_i(X_i)$ by at most $1$ (in fact, it is exactly $1$, for otherwise the first update rule should be applied), does not make $i$ envy $k$.

Now, consider any two agents $i,j\in S$.
If $(i,j)\notin E$ before the update, $i$ will not envy $j$ after the update even in the worst case where $v_i(X_i)$ is increased by $1$ and $v_i(X_j)$ is unchanged.
If $(i,j)\in E$ before the update, $(i,j)$ is on a cycle by the property of strongly connected components.
Since the first two update rules do not apply, it must be that $c_i(e\mid X_i)=c_i(e\mid X_j)=1$ for any unallocated item $e$.
Adding an item to $i$ and an item to $j$ increases both $c_i(X_i)$ and $c_i(X_j)$ by $1$.
Envy-freeness is again preserved.

Finally, checking that the algorithm runs in polynomial time is straightforward.
We can only reach a partial allocation when the if-condition at Line~12 fails, in which case the number of unallocated items is less than $|S|$, which is at most $n-1$.

\newpage
\bibliography{aaai24}

\newpage
\appendix	

\section{Relationship Between Set Functions with Binary Marginals}
\label{sect:relationship}
In this section, we explore the relationship between additive, cancelable, and submodular set functions, and we focus only on set functions with binary marginals.
Many results will be used in later sections, and they are interesting observations that are independent of our applications to fair division.

Firstly, it is obvious from definitions that all additive functions are cancelable.
Moreover, there exist cancelable set functions with binary marginals that are not additive.
An example is given below.
\begin{equation}\label{eqn:cancelable}
    \phi(X)=\left\{\begin{array}{cl}
        5 & \mbox{when }|X|\geq 5 \\
        |X| & \mbox{otherwise}
    \end{array}\right..
\end{equation}

Therefore, we have the following theorem.
\thmAddCancelable*

Next, we show that, by applying binary marginal property to cancelable functions, any cancelable function with binary marginals is also submodular.

\thmCancelableSubmodular*
\begin{proof}
    To show the set containment, we will show that any set function $\phi:M\to\mathbb{Z}_{\geq0}$ (with binary marginals) that is not submodular cannot be cancelable.
    Since $\phi$ is nonsubmodular, there exists $S\subseteq T\subseteq M$ and $e\in M\setminus T$ such that $\phi(S+e)-\phi(S)<\phi(T+e)-\phi(T)$.
    We initialize $T'=S$ and iteratively add an element from $T\setminus S$ to $T'$.
    At the start, we have $\phi(S+e)-\phi(S)=\phi(T'+e)-\phi(T')$.
    We consider the first time we observe $\phi(S+e)-\phi(S)<\phi(T'+e)-\phi(T')$ after an element $f$ is added to $T'$.
    Let $U$ be the state of $T'$ before $f$ is added.
    Since $\phi$ has binary marginals, we must have $\phi(S+e)-\phi(S)=\phi(U+e)-\phi(U)$ and $\phi(S+e)-\phi(S)<\phi(U+e+f)-\phi(U+f)$.
    Combining the two equations and by noting that $\phi$ is monotone, we must have $0\leq\phi(U+e)-\phi(U)<\phi(U+e+f)-\phi(U+f)$.
    Since $U+e+f$ and $U+f$ differ by only one element and $\phi$ has binary marginals, we must have $0\leq\phi(U+e)-\phi(U)<\phi(U+e+f)-\phi(U+f)\leq 1$.
    It must be that $\phi(U+e)-\phi(U)=0$ and $\phi(U+e+f)-\phi(U+f)=1$.
    Thus, $\phi(U+e+f)>\phi(U+f)$ fails to imply $\phi(U+e)>\phi(U)$, so $\phi$ is not cancelable.

    To show the containment is proper, we present an example of submodular set functions (with binary marginals) that is not cancelable.
    Let $M=\{a,b,c,d\}$ and
    $$\phi(X)=\left\{\begin{array}{cl}
        |X|-1 & \mbox{if }\{a,b,c\}\subseteq X \\
        |X| & \mbox{otherwise}
    \end{array}\right..$$
    It is straightforward to check (e.g., by enumerating all cases) that $\phi$ is submodular.
    It is not cancelable as $\phi(\{c,d\})=\phi(\{b,c\})=2$ and $\phi(\{a,c,d\})=3>2=\phi(\{a,b,c\})$.
\end{proof}

Note that for set functions with non-binary marginals, the containment in Theorem~\ref{thm:cancelable-submodular} does not hold.
A simple cancelable set function that is non-submodular is $\phi(X)=2^{|X|}$.

We list some of the useful properties of cancelable functions, whose proofs are straightforward.

\propcancelable*

Finally, we prove some properties for submodular functions.
By Theorem~\ref{thm:cancelable-submodular}, cancelable functions with binary marginals also satisfy these properties.

\propsubmodular*
\begin{proof}
    1 is a well-known alternative definition for submodularity.
    For 2, if $\phi(S-e)=0$ for any $e\in S$, this means every subset of $S$ with size $|S|-1$ has value $0$.
    By monotonicity of $\phi$, every element of $S$ has value $0$.
    Then, $\phi(S)=1$ violates the submodularity. 
\end{proof}
%    \begin{lemma}
%        For any cost function $c_i$ with binary marginal, it is cancelable if and only if it is submodular.
%    \end{lemma}

\section{Missing Proofs in Section~\ref{sec:additive}}\label{app:additive-missing}
\begin{proofof}{Lemma~\ref{lemma:efx-additive}}
    We show that the allocation $\bX$ is EFX by mathematic induction.
    At the beginning of Phase 2, we must have that the partial allocation is EFX since each agent only receives items that cost $0$ to her.
    In the following, we assume that the (partial) allocation is EFX at the beginning of some round $t$.
    We show that the (partial) allocation is still EFX for all agents at the end of round $t$.
    
    According to the algorithm, we only have to consider the case under the if condition (in line $9$) that we reallocate some items.
    Let $X_i, X_j$ be the bundles that agents $i, j$ hold at the beginning of round $t$, respectively.
    During the round, we reallocate a subset of items $S\subseteq X_j$ to agent $i$ and assign an item $e\in M^+$ to agent $j$.
    We must have $c_i(e) = 1$ otherwise the if-condition does not hold.
    We show that at the end of the round, the new allocation $\bX'$ is EFX for all agents, while $X'_i = X_i \cup S, X'_j = X_j \setminus S + e$ and $X'_k = X_k$ for all $k\in N\setminus \{i,j\}$.
    Before we give the proof, we first show the following claim.
        \begin{claim}\label{claim:all-1}
            If the if-condition from Line 9 to Line 12 is executed, we must have $c_i(X_i) = c_j(X_j)$.
            After the execution, we have $X_j \setminus S \subseteq M^+$.
        \end{claim}
        \begin{proof}
            For the first statement, we assume otherwise that $c_j(X_j) \geq c_i(X_i) + 1$.
            Then we have $c_i(X_i+e) = c_i(X_i) + 1 \leq c_j(X_j) \leq c_i(X_j)$, where the last inequality holds since the (partial) allocation minimizes social cost.
            In other words, agent $i$ is envy-free towards $j$ after allocating $e$ to $i$, which is a contradiction.
            Hence the only case that agent $i$ is not EFX towards $j$ after assigning item $e$ is that $c_i(X_i) = c_j(X_j) = c_i(X_j)$.
            Following the minimum social cost property, for any item $e'\in X_j$, we have $c_i(e') = c_j(e')$.
            Hence we have $S\subseteq M^0$ and $X_j \setminus S \subseteq M^+$.
        \end{proof}

    Given Claim~\ref{claim:all-1}, we are ready to show the allocation is EFX for all agents at the end of round $t$.
    We show the property holds for agents $i, j$ and any $k\neq i,j$ individually. 
        \begin{itemize}
            \item For any $k \neq i, j$, agent $i$ is EFX towards $j$ and $k$.
            Note that during the reallocation, we only reallocate those items that cost $0$ to agent $i$, i.e., $c_i(S) = 0$.
            Hence we have $c_i(X'_i) = c_i(X_i \cup S) = c_i(X_i)$.
            Note that the allocation $\bX$ is EFX for agent $i$ and $X_k = X'_k$ for any $k\neq i,j$.
            We have agent $i$ is EFX towards any agent $k$ at the end of the round.
            As for the envy between $i$ and $j$, we have $c_i(X'_j) = c_i(X_j \setminus S + e) = c_i(X_j) + 1$, agent $i$ is EFX towards agent $j$ at the end of the round.
            
            \item For any $k\neq i, j$, agent $j$ is EFX towards $i$ and $k$.
            Following Claim~\ref{claim:all-1}, agent $j$ has the minimum bundle cost among all agents at the beginning of the round.
            In other words, agent $j$ is envy-free towards any agent $k \neq j$ since $c_j(X_j) \leq c_k(X_k) \leq c_j(X_k)$, where the second inequality holds since the (partial) allocation minimizes social cost.
            Combining $X_j \setminus M^+$ and $e\in M^+$, we have $c_j(X'_j - e') = c_j(X_j) \leq c_j(X_k)$ for any $e' \in X'_j$ and any $k\neq j$.
            
            \item For any $k \neq i, j$, agent $k$ is EFX towards $i$ and $j$.
            Due to the monotonicity of $c_k$, we have $c_k(X'_i) \geq c_k(X_i)$, agent $k$ is EFX towards $i$ at the end of the round.
            Next, we show that $k$ is EFX towards $j$.
            We first claim that $c_k(X_k) \leq c_i(X_i) + 1$.
            Assume otherwise $c_k(X_k) \geq c_i(X_i) + 2$, we consider the last time that we assign an item $e' \in M^+$ to agent $k$.
            We must have $c_k(X_k - e') \geq c_i(X_i) + 1$, which contradicts the fact that $k$ holds a bundle with minimum cost.
            Hence we consider the cases that $c_k(X_k) = c_i(X_i)$ and $c_k(X_k) = c_i(X_i) + 1$.
            For both cases $k$ is EFX towards $j$ since $c_k(X_k) \leq c_i(X_i) + 1 = c_j(X'_j) \leq c_k(X'_j)$, where the equality follows from Claim~\ref{claim:all-1} and the fact that $e\in M^+$.
        \end{itemize}
    In conclusion, we show that the (partial) allocation is EFX for all agents, at the end of the round.
    Note that in each round we allocate one item in $M^+$ and the algorithm terminates Phase 2 after $|M^+|$ rounds.
    Hence upon the running of Algorithm~\ref{alg:additive}, it computes an EFX allocation.
\end{proofof}

\begin{lemma}\label{lemma:poly-additive}
    Algorithm~\ref{alg:additive} runs in $O(nm^2)$ time.
\end{lemma}
\begin{proof}
    Algorithm~\ref{alg:additive} starts from a partial allocation $\bX^0$ with zero social cost, which can be computed in $O(mn)$ times (there are at most $m$ items in $M^0$ and each item can be allocated in $O(n)$ time).
    The algorithm runs in rounds in phase 2 and there are at most $m$ rounds to allocate items in $M^+$.
    In each round, determining the agent $i$ who receives item $e$ can be done in $O(\log n)$ time, and determining whether $i$ is EFX towards other agents can be done in $O(mn)$ time.
    The reallocation in each round can be done in $O(m)$ time since there are at most $m$ items in $X_j$.
    In conclusion, Algorithm~\ref{alg:additive} runs in $O(nm^2)$ time.
\end{proof}

\subsection{Non-existence of EFX and PO}\label{ssec:non-existence}
In this section, we complete our result by exploring the non-existence of EFX and PO allocations in the additive setting.
We show that even extending our result to ternary instances (where the cost of each item can only be $0,1,2$) is impossible.
Whether bivalued instances (another generalization of binary instances) admit EFX and PO allocations for the general number of agents remains open.

    \begin{definition}[Ternary Instances]
        An instance is called \emph{ternary} if for each agent $i\in N$ and any item $e\in M$ we have $c_i(e) \in \{0,1,2\}$.
    \end{definition}
    
    \begin{example} \label{example:hard-trinary}
    Consider an instance with $n = 2$ agents and $m = 3$ items.
    Both agents have ternary cost functions, that is, $c_i(e) \in \{0,1,2\}$ for all $i\in N $ and $e\in M$.
        The costs are shown in Table~\ref{tab:hard-trinary}.
        Note that any PO allocation should assign $e_2$ to agent $2$ and $e_3$ to agent $1$.
        Hence PO allocations can only be $X_1 = \{e_1, e_3\}, X_2 = \{e_2\}$ or $X_1 = \{e_3\}, X_2 = \{e_1. e_2\}$.
    None of these PO allocations is EFX since the agent who received item $e_1$ still envies another agent even after removing the item with zero cost.
    \begin{table}[htbp]
        \centering
        \begin{tabular}{c|c|c|c}
            &  $e_1$ & $e_2$ & $e_3$ \\ \hline
            agent 1   & $2$ & $1$ & $0$ \\
            agent 2   & $2$ & $0$ & $1$
        \end{tabular}
        \caption{Instance showing that EFX and PO allocations do not exist for ternary chores.}
        \label{tab:hard-trinary}
    \end{table}
\end{example}

\section{Missing Proofs in Section~\ref{sect:cancelable}}
\label{sect:missingproofcancelable}
\begin{proofof}{Lemma~\ref{lem:phase1}}
    We prove by induction that $c_i(A_j)=t$ for each pair of $i$ and $j$ after $t$ while-loop iterations.
    The base step for $t=0$ is trivial.
    Suppose the claim holds for $t$.
    We will show it holds for $t+1$.
    Let $(A_1,\ldots,A_n)$ be the allocation before the $(t+1)$-th iteration.
    By induction hypothesis, $c_i(A_j)=t$ for every pair of $i$ and $j$.
    For each item $e$ allocated in the $(t+1)$-th iteration, we have $c_i(e\mid A_i)=1$ for each agent $i$.
    By Property~1 in Proposition~\ref{prop:cancelable}, this also implies $c_i(e\mid A_j)=1$ for every other agent $j$.
    Thus, $c_i(A_j)=t+1$ for every pair of $i$ and $j$ after the $(t+1)$-th iteration.
\end{proofof}

\begin{proofof}{Theorem~\ref{thm:cancelable-po}}
    Consider $n$ agents with $5n$ items where each agent's cost function is given by Equation~(\ref{eqn:cancelable}).
    It is easy to check that the only EFX allocation gives exactly $5$ items to each agent.
    However, this allocation is Pareto-dominated by the allocation that allocates all items to a single agent.
\end{proofof}

\section{Submodular Cost Functions with Binary Marginals}\label{sec:submodular}
\begin{theorem}\label{thm:submodular}
    For submodular cost functions with binary marginals, there exists an allocation $\bX$ that is either EFX or $2$-EF.
    In addition, such an allocation can be computed in polynomial time.
\end{theorem}

Notice that a $2$-EF allocation is always $2$-EFX. We have the following corollary.
\begin{corollary}\label{cor:submodular}
    For submodular cost functions with binary marginals, there exists a $2$-EFX allocation and it can be computed in polynomial time.
\end{corollary}

The proof of Theorem~\ref{thm:submodular} is mostly based on algorithms in the previous sections.

\begin{proof}[Proof of Theorem~\ref{thm:submodular}]
Let $M_1=\{e\in M\mid c_i(e)=1\mbox{ for all }i\}$.
We consider two cases.

Case 1: $|M_1|<n$. In this case, we will run Algorithm~\ref{alg:cancelable}.
Note that Phase~1 will be skipped.
As a result, each $d_i$ in the algorithm will be $c_i$, which is submodular.
By Lemma~\ref{lem:phase2} and Remark~\ref{remark:phase2}, an EFX allocation will be output.

Case 2: $|M_1|\geq n$.
In this case, we will initialize the allocation $(X_1,\ldots,X_n)$ such that each $X_i$ contains exactly one item in $M_1$.
The current partial allocation is clearly EF.
Next, we will implement the while-loop in Algorithm~\ref{alg:general}.
Lastly, for the unallocated items, we allocate them arbitrarily such that each agent receives at most one extra item.

Since we have $c_i(X_j)=1$ for any $i$ and $j$ before the while-loop and the algorithm never removes items from a bundle, it holds that $c_i(X_j)\geq 1$ for any $i$ and $j$ for the final allocation.
The allocation remains EF during the while-loop, and the envy-freeness can only be broken at the last step.
By the binary marginal property, if $i$ envies $j$, it must be that $c_i(X_i)=w+1$ and $c_i(X_j)=w$ for some $w$.
We have $c_i(X_i)\leq 2\cdot c_i(X_j)$ since we have seen that $w\geq 1$.
\end{proof}
\end{document}